\documentclass[]{article}
\usepackage{amsmath,amsfonts,amssymb}
\usepackage{amsthm}
\usepackage[mathfrak]{}
\usepackage{cases}
\usepackage[latin1]{inputenc}
\usepackage[T1]{fontenc}
\usepackage{verbatim}
\usepackage{graphicx}
\usepackage{float}
\usepackage{mathrsfs}
\usepackage [all]{xy}
\usepackage{color}
\usepackage{tikz}
\usepackage{comment}
\usepackage{hyperref}
\usepackage{cases}

\newtheorem{thm}{Theorem}[section]
\newtheorem{lem}[thm]{Lemma}
\newtheorem{cor}[thm]{Corollary}

\newtheorem{defn}[thm]{Definition}
\newtheorem{exmp}[thm]{Example}
\newtheorem{rem}[thm]{Remark}

\newtheorem{prob}[thm]{Problem}

\title{On the solutions of linear systems over additively idempotent semirings.}
\author{\'Alvaro Otero S\'anchez \footnote{Department of Mathematics of the University of Almer\'ia}, Daniel Camaz\'on \footnote{Department of Mathematics of the University of Almer\'ia} \footnote{The second author is partially supported by Ministerio de Ciencia e Innovaci\'on  PID2022-138906NB-C21.}, Juan Antonio L\'opez Ramos \footnote{Department of Mathematics of the University of Almer\'ia} \footnote{The third author is supported by Ministerio de Ciencia e Innovaci\'on PID2020-113552GB-I00 and FQM 0211 Junta de Andaluc\'ia.}}
\date{aos073@inlumine.ual.es}

\begin{document}
\maketitle

\begin{abstract}
The aim of this article is to solve the system $XA=Y$ where $A=(a_{ij})\in M_{m\times n}(S)$, $Y\in S^{m}$ and $X$ is an unknown vector of size $n$, being $S$ an additively idempotent semiring. If the system has solutions then we completely characterize its maximal one, and in the particular
case where $S$ is a generalized tropical semiring a complete characterization of its solutions is provided as well as an explicit bound of the computational cost associated to its computation. Finally, when $S$ is finite, we give a cryptographic application by presenting an attack to the key exchange protocol proposed by Maze, Monico and Rosenthal.
\end{abstract}

\section{Introduction}

A semiring $(S,+,\cdot,0,1)$ is an algebraic structure in which $(S,+)$ is a commutative monoid with an identity element $0$ and $(S,\cdot)$ is a monoid with an identity element $1$, being both internal operations connected by ring-like distributivity. The additive identity $0$ is multiplicatively absorbing, and $0\neq 1$ (see, e.g., the monograph \cite{Golan99} for an intensive treatment of this algebraic structure). Moreover, a semiring $(S,+,\cdot,0,1)$ is said to be additively idempotent if $x+x=x$ for all $x\in S$. Historically, the first notion of a semiring was due to Vandiver \cite{Vandiver34} in $1934$  and the interest in additively idempotent semirings arose in the $1950s$ through the observation that some problems in discrete optimization could be linearized over such structures (see, e.g., \cite{Cuninghame79}). The first work to make use of an algebra over an idempotent ring (apart from Boolean fields) was that of Kleene \cite{Kleene56} where nerve nets were studied in the context of finite state machines. Since then, the study of additively idempotent semirings has lead to multiple connections with such diverse fields as, e.g., graph theory (path algebra), Hamilton-Jacobi theory, automata and language theory, discrete event system theory (where linear systems over additively idempotent semirings modelize discrete event systems of practical interest), and fuzzy logic. As some examples of connections with the latter one, each fuzzy triangular norm (t-norm, see e.g. \cite{Klement13}) conducts to an additively idempotent semiring, which is called in the literature a max-t semiring, and in \cite{Nola05}, \cite{Nola07}, \cite{Nola13}, Nola et al. study certain objects of algebra over semirings arising from fuzzy logic such as MV-algebras or the Lukasiewicz transform. Moreover, nowadays there is a vast literature on matrices with idempotent coefficients and their applications e.g., \cite{Cuninghame79}, \cite{Krivulin09}, \cite{Rodionov11}. \\

As an example of an additively idempotent semiring, we will study the tropical semiring. Tropical algebra was the first section of tropical mathematics to appear, and although a systematic study of the tropical semiring began only after the works of Simon (see \cite{Simon78}), we should note that the $(\mathbb{R},min,+)$ semiring had appeared before in optimization problems (see, e.g., Floyd's algorithm for finding shortest paths in a graph \cite{Floyd62}).\\
Although the problem of solving linear systems was formulated right after the definition of a root for a tropical polynomial was given by Viro \cite{Viro01}, the first paper \cite{Sturmfels05} actually devoted to tropical linear algebra appeared only as late as in $2005$. Moreover, this problem has already proved to be very interesting from the algorithmic point of view as it is known to be in $NP\cap coNP$. Some examples of algorithms proposed for solving tropical linear systems can be found in \cite{Grigoriev13}, \cite{Davydow16} and \cite{Olia21}. At present, there are numerous applications of linear systems over tropical semirings in various areas of mathematics, engineering and computer science. For instance, Noel, Grigoriev, Vakulenko, and Radulescu have recently proposed a way to use algorithms for solving tropical linear systems to study stable states of reaction networks in biology \cite{GrigorievI12} \cite{GrigorievII12}. As an application in fuzzy set theory, Gavalec, N\v emcová and Sergeev have recently proposed a way to convert the problems of max-Lukasiewicz linear algebra, i.e., the linear algebra over max-Lukasiewicz semiring, to the problems of tropical (max-plus) linear algebra \cite{Sergeev15}, and take advantage of the well-developed theory and algorithms of the latter in order to develop a theory of the matrix powers and the eigenproblem over the max-Lukasiewicz semiring. Thus, problems of tropical linear algebra and in particular tropical linear systems are important from both theoretical and practical implications. \\

Letting $(S,+,\cdot)$ be an additively idempotent semiring, we want to solve the system $XA=Y$, where $A=(a_{ij})\in M_{m\times n}(S)$, $Y\in S^{m}$ and $X$ is an unknown vector of size $n$. We have to clarify that our notion of solution differs from that of Viro, in the sense that the maximum is achieved only once. If the system $XA=Y$ has solutions, then we can completely characterize its maximal one. Moreover, in the particular case where $S$ is a generalized tropical semiring (see Definition 1.1.1) we are able to characterize completely its solutions and give an explicit bound of the computational cost associated to its computation. Finally, we give a cryptographic application by applying our previous results to the case of $S$ finite, and propose an attack to the key exchange protocol presented in \cite{Rosenthal07}.

\section{Additively idempotent semirings}

We start by recall some basic background and the notation followed through this work. 

\begin{defn}
    A semiring $R$ is a non-empty set together with two operations $+$ and $\cdot$ such that $(S,+)$ is a commutative monoid, $(S,\cdot)$ is a monoid and 
    the following distributive laws hold:

    \begin{equation}
        \begin{split}
            a(b+c) & = ab + ac \\
            (a+b)c & = ac + bc
        \end{split}
    \end{equation}
    
    We say that $(R,+,\cdot)$ is additively idempotent if $a+a=a$ for all $a \in R$.
\end{defn}

\begin{defn}
    Let $R$ be a semiring and $(M,+)$ be a commutative semigroup with identity $0_M$. $M$ is a right semimodule over $R$ if there is an external operation  $\cdot \ : M\times R\rightarrow M$ such that

    \begin{equation}
    \begin{split}
             (m\cdot a)\cdot b & = m\cdot (a\cdot b)\\
             m\cdot (a+b) & = m\cdot a + m\cdot b \\
             (m+n)\cdot a & = m\cdot a +n\cdot a \\
             0_M\cdot a & = 0_M
    \end{split}
    \end{equation}

    for all $a, b \in R$ and $m,n \in M$. We will denote  $m\cdot a$ by simple concatenation $ma$. 
\end{defn}

Let $(R,+,\cdot)$ be an additively idempotent semiring. Every such a semiring is endowed with an order given by the first operation, which is defined as

\begin{equation}
    a\leq b \mbox{ if and only if } a+b=b.
\end{equation}

This order respects the operation in $R$ and enables as to define a partial order in $R^n$ for every positive integer $n$.

\begin{equation}
    X=(x_1,\dots x_n) \geq Y = (y_1,\dots, y_n) \mbox{ if and only if } x_i \geq y_i \text{   } \forall i=1, \dots, n.
\end{equation}

If $R$ is a semiring, then we will denote by $Mat_n(R)$ the semiring of square matrices of order $n$ for some positive integer $n$ and whose entries are in $R$.

\begin{lem} \label{RepetaOper}
    The previous order is compatible with the operations in $R^n$ as right $Mat_n(R)-$semimodule.
\end{lem}
\begin{proof}

On one hand, if $X=(x_1,\dots ,x_n), \ Y = (y_1,\dots, y_n)  \in R^n$ are such that $X\geq Y $ and $C=(c_1,\dots c_n) \in R^n$, then we have that
$X\geq Y$  implies that  $x_i \geq y_i \ \forall i \in \{1,\dots,n\}$ and therefore  $x_i + c_i \geq y_i +c_i \ \forall i  \in \{1,\dots,n\}$. Thus 
$X+C\geq Y+ C$.

%The proof for $X\geq Y \Rightarrow C+X\geq C+Y$ is analogous to the previous one.

On the other hand, if  $A=(a_{i,j})\in Mat_n(R)$, then  $X\geq Y$, i.e. $x_i \geq y_i \ \forall i=1,\dots,n$. Then $x_ia_{i,j}   \geq y_i a_{i,j} \ \forall i,j  \in \{1,\dots,n\}$ and thus 
$\sum_i x_i a_{i,j} \geq \sum_i y_i a_{i,j} \ \forall i,j  \in \{1,\dots,n\}$. Therefore $XA \geq YA$.

\end{proof}

Let $XA=Y$ be the system of linear equations in $R$ with indeterminates $x_1, \dots , x_n$,

\begin{equation}
    x_1 \begin{pmatrix}
a_{11} \\
a_{12} \\
\vdots \\
a_{1(m-1)} \\
a_{1m} \\
\end{pmatrix}  
+
x_2 \begin{pmatrix}
a_{21} \\
a_{22} \\
\vdots \\
a_{2(m-1)} \\
a_{2m} \\
\end{pmatrix} 
+\dots
+ x_n \begin{pmatrix}
a_{n1} \\
a_{n2} \\
\vdots \\
a_{n(m-1)} \\
a_{nm} \\
\end{pmatrix}
=\begin{pmatrix}
y_{1} \\
y_{2} \\
\vdots \\
y_{m-1} \\
y_{m} \\
\end{pmatrix},
\end{equation}

with $a_{i,j},y_{j} \in R$ for all $i=1, \dots ,n$ $j=1,\dots ,m$. If we denote by $A_i$ the $i-th$ row of $A$,  $A_i=(a_{i1},a_{i2},\dots  ,a_{im} )$, then we have that the system can be written as

\begin{equation}
    x_1A_1 +x_2A_2 + \dots + x_nA_n = Y.
\end{equation}

\begin{defn}
Let $R$ be an additively idempotent semiring, and let $XA=Y$ be a linear system of equations. We say that $\hat{X}$ is the maximal solution of the system if and only if the two following conditions are satisfied 
    \begin{enumerate}
        \item $\hat{X}\in R^n$ is a solution of the system, i.e. $\hat{X}A=Y$,
        \item if $Z\in R^n$  is any other solution of the system, then $Z+\hat{X}=\hat{X}$.
    \end{enumerate}
This last condition is equivalent to $Z\leq \hat{X}$.
\end{defn}

\begin{thm} \label{SolucionSistemaGeneral}
Given $(R,+,\cdot)$ an additively idempotent semiring, let $W_i = \{ x \in R : x A_i + Y = Y \}$ $\forall i =1,\dots,n$. Suppose that these subsets have a maximum with respect to the order induced in $R$

    \begin{equation}
        C_i = \max W_i.
    \end{equation}

\noindent If $XA = Y$ has as a solution, then $Z=(C_1,\dots C_n)$ is the maximal solution of the system.
\end{thm}

\begin{proof}
If there is a solution $\hat{X} = (\hat{x}_1,\dots , \hat{x}_n)$, then, for all $k = 1, \dots, n$ we have that
    
    \begin{equation} \label{condicion1}
        \hat{X}A= Y  \Rightarrow \hat{x}_1A_1 +\hat{x}_2A_2 + \dots \hat{x_k}A_k + \dots + \hat{x}_nA_n = Y \Rightarrow \hat{x}_kA_k + Y = Y \Rightarrow \hat{x}_k \in W_k,
    \end{equation}

where we used the following relation

\begin{align*}
        Y & = Y + Y, \\  
				  & = \hat{x}_1 \cdot A_1  + \hat{x}_2 \cdot A_2 + \dots + \hat{x}_k \cdot A_k + \dots +\hat{x}_n \cdot A_n + Y, \\
					& = \hat{x}_1 \cdot A_1  + \hat{x}_2 \cdot A_2 + \dots + \hat{x}_k \cdot A_k + \hat{x}_k \cdot A_k + \dots +\hat{x}_n \cdot A_n + Y, \\
					& = \hat{x}_k \cdot A_k + \hat{x}_1 \cdot A_1  + \hat{x}_2 \cdot A_2 + \dots \hat{x}_n \cdot A_n  + Y, \\  
					& = \hat{x}_k \cdot A_k + Y.
    \end{align*}

    Since $\hat{x}_k \in W_k$ from \ref{condicion1}, then we have $C_k\geq \hat{x_k}$ $\forall k = 1, \dots, n$, and hence, by the proof of Lemma \ref{RepetaOper}

    \begin{equation} \label{eqSolMax}
        Z \geq \hat{X} \Rightarrow ZA \geq  \hat{X}A = Y.
    \end{equation}

    In addition, as $max W_i \in W_i$, by the definition of $W_{i}$, we get that

    \begin{equation}
        C_{i} \in W_{i} \Rightarrow ZA \leq Y,
    \end{equation}

 and thus, $ZA = Y$, i.e., $Z$ is a solution. Furthermore, by definition of the order in $R^n$, we have that this solution is maximal.
\end{proof}

\section{Tropical semirings}

As we have shown in Theorem \ref{SolucionSistemaGeneral} we can characterize the maximal solution of a linear system over an additively idempotent semiring under some circumstances. Our aim in this section is to study the existence of solutions in the particular case of tropical semirings. 

\begin{defn}
    Let $(R,+,\cdot)$ be a semiring. We say that $R$ is a generalized tropical semiring if
    \begin{center}
     $a+b = a$ or $a+b=b$ of all $a,b \in R$.
     \end{center}
  \end{defn}

The following lemma is immediate from the preceding definition. 

\begin{lem}
    Every generalized tropical semiring is totally ordered with respect to the order induced by the addition.
\end{lem}

\begin{exmp}
$(\mathbb{N},\max,\cdot)$ is a semirings where $a+b=\max\{a,b\} = a \text{ or } b$, and thus they are generalized tropical semirings. Analogously, $(\mathbb{R},\max,+)$, $(\mathbb{Z},\max,+)$ and $(\mathbb{Q},\max,+)$ are also generalized tropical semirings and they verify being a group with respect to the second operation. 
\end{exmp}

The previous example induces the following definition.

\begin{defn}
    Let $(S,+)$ a semigroup with a total order which is compatible with the operation $+$. We define the tropicalized semiring of $S$ as the semiring $Trop(S)=S \cup \{ \infty \}$  with inner addition defined by $\max$, given by the order in $S$, and inner product defined by $+$, the inner operation of $S$, extending these to $\infty$ in the following way: 
    
    \begin{enumerate}
        \item $a+\infty = \infty + a = \infty$ $\forall a \in Trop(S)$.
        \item $\max\{a,\infty \} = \max\{\infty,a\} = \infty$ $\forall a \in Trop(S)$.
    \end{enumerate}

\end{defn}

The following result is straightforward.

\begin{lem}
    Let $(S,+)$ be a totally ordered semigroup, and let $(Trop (S), \max , + )$ be its tropicalized semiring. Then $(Trop (S), \max , + )$ is a generalized tropical semiring.
\end{lem}

Let us recall from \cite{Olia21} that the tropical semiring is given by semiring $(\mathbb{R} \cup \{\infty\},\max,+)$. It is immediate that the tropical semiring is the tropicalized of $\mathbb{R}$ with the usual operations.

\begin{thm} \label{solucionSemianilloTropical}
    Let $(R,+, \cdot)$ be a generalized tropical semiring where $(R,\cdot)$ is a group. Then the linear system $ X\cdot A = Y$ has at least one solution.
\end{thm}

\begin{proof}
    Firstly, we will prove that the sets $W_i = \{ x \in R : x \cdot A_i + Y = Y \}$  with $A_i = (a_{i,j})_{j=1\dots,m}$, $i=1,\dots,n$ have maximum so we could use Theorem \ref{SolucionSistemaGeneral}.

    If $x\in W_{i}$, then we have that $x \cdot A_i + Y = Y$, where, if we see the row $j-th$, we get $x \cdot a_{i,j} + y_j = y_j$. Therefore
    
    \begin{equation}
        \max \{x \cdot a_{i,j}, y_j \} = y_j \Rightarrow x \cdot a_{i,j} \leq y_j \Rightarrow x  \leq y_j \cdot a_{i,j}^{-1},
    \end{equation} 
    
    so $x \in W_{i}$ if and only if $x \leq y_j \cdot a_{i,j}^{-1}$ for all $j \in \{1,\dots,m\}$. This condition is verified if $x \leq \min_{j} \{y_j \cdot a_{i,j}^{-1} \}$. 
    
    Now, if we denote $C_i = \min_{j} \{ y_j \cdot a_{i,j}^{-1} \}$, we get that $C_i$ is an upper bound of $W_i$, because $x \in W_{i} \Rightarrow x \leq C_i$, and, in addition, it belongs to the set $W_i$, due to the following identity
    
    \begin{equation}
        C_i a_{i,j} = \min_{j} \{ y_j \cdot a_{i,j}^{-1} \} a_{i,j} \leq y_j a_{i,j}^{-1} a_{i,j} = y_j,
    \end{equation}
    
    that holds for all $j \in \{1,\dots,m\}$, where $C_i a_{i,j}+ y_j = y_j$. We can conclude that $\max W_i = C_i$.
\end{proof}

Moreover, as a consequence of the previous result, we get the following corollary.

\begin{cor} \label{caracterizacionMaximo}
    Let $(R,+, \cdot)$ be a generalized tropical semiring where $(R,\cdot)$ is a group, $A=(a_{i,j}) \in Mat_{n\times m} (R)$ and the column vector $Y=(y_1, \dots ,y_m)\in R^m$. Then $(M_1,\dots ,M_n)$ is a solution of the linear system $XA=Y$, where $M_i= \min_{j} \{ y_j a_{i,j}^{-1} \} = \max W_i$ for $i=1,...,n$.
\end{cor}

We can also point out that in case the semiring $(R,+,\cdot)$ is such that $(R,\cdot)$ is not a group, we can use the following theorem from \cite{SemigroupIntoGroup}.

\begin{thm}
    A commutative semigroup can be embedded in a group if and only if it is cancellative.
\end{thm}

\begin{thm}\label{ThmEmb}
    Every generalized tropical semiring $(R,+,\cdot)$ such that $(R,\cdot)$ is cancellative can be embedded into a generalized tropical semiring having inverses with respect to $\cdot$.
\end{thm}

\begin{proof}
    Let $(R,+,\cdot)$ be a generalized tropical semiring such that $(R,\cdot)$ is cancellative. Then it can be embedded into a group, which we will denote as $Q(R)$. Note that the elements of  $Q(R)$ are of the form $a/b := ab^{-1}$ with $a,b \in R$. As $R$ is totally ordered, then we can endow $Q(R)$ with an order by 
\begin{equation}
    \frac{a}{b} \leq \frac{c}{d} \Leftrightarrow ad\leq bc \text{ }\forall a,b,c,d \in R.
\end{equation}

    Moreover, as the order in $R$ is total, the order induced in $Q(R)$ is also a total order. We can define the addition in $Q(R)$ as

\begin{equation}
    \frac{a}{b} +_Q \frac{c}{d} = \max\Big\{\frac{a}{b} , \frac{c}{d} \Big\}.
\end{equation}

Note that due to the properties of $\max$, we have that  $Q(R)$ is a generalized tropical semiring, where there are inverses respect to the second operation.

Finally, $R \hookrightarrow Q(R)$ is an injective homomorphism of semirings.
\end{proof}

\begin{exmp}
    We have that $(\mathbb{N},\max, \cdot)$ is a generalized tropical semiring. Furthermore, $(\mathbb{N},\cdot)$ is cancellative. By the previous result, we can embed $(\mathbb{N},\max, \cdot)$ into a generalized tropical semiring with inverses, which by the preceding construction, can be $(\mathbb{Q}_{>0},\max, \cdot)$.
\end{exmp}

Now we will show how to find every solution of the previous system.

\begin{lem} \label{CiWi}
 Let $R$ be a generalized tropical semiring, and $XA=Y$ a linear system of equations, being $A=(a_{i,j}) \in Mat_{n\times m} (R)$ and  $Y=(y_1, \dots ,y_m)\in R^m$. Let $C_i = \max W_i $. Then $x \in W_i$ if and only if  $x\leq C_i$.
\end{lem}

\begin{proof}
    
Let $x\leq C_i$ and let $A_i$ the $i$-th row of $A$, for $i=1, \dots ,n$. Then $x A_i\leq C_i A_i$ and so $x A_i + Y\leq C_i A_i + Y = Y$ Moreover, $Y \leq xA_i + Y$, since

 \[ Y + (xA_i + Y) = (xA_i + Y )+ Y = xA_i + (Y+Y) = xA_i + Y\] 
 
 and hence $xA_i + Y=Y$ and $x_i \in W_i$.
\end{proof}

Let $R$ be a generalized tropical semiring, and let $XA=Y$ be a linear system of equations with $Y=(y_i)\in R^m$ and $A=(a_{i,j}) \in Mat_{n\times m} (R)$. Let $W_i = \{ x \in R : x \cdot A_i + Y = Y \}$, then by the proof of Theorem \ref{solucionSemianilloTropical}, $W_i$ has a maximum that will be denoted by $C_i$ for all $i=1,\dots,n$.

\begin{thm}\label{ThmChar}
Let $R$ be a generalized tropical semiring, and let $XA=Y$ be a system of equations with $Y=(y_i)\in R^m$ and $A=(a_{i,j}) \in Mat_{n\times m} (R)$. $X = (x_1, x_2,...,x_n)$ is solution of the system if and only if 

\begin{enumerate}
    \item $x_i \cdot a_{i,j} + y_{j} = y_{j}$ ,$\forall j =1,\dots,m$,\label{unoLatino}
    \item $\forall j =1,\dots,m$ $\exists h \in \{1,\dots,n\}$ such that $x_h\cdot a_{h,j} = y_{j}$ \label{dosLatino}.
\end{enumerate}  

\end{thm}

\begin{proof}
Let us assume first that $X = (x_1, x_2,...,x_n)$ is solution of the system. Then the first condition has been already proven in  equation \ref{condicion1}. Let us show now the second condition. We have that

    \begin{equation}
         x_1 \cdot A_1  + x_2 \cdot A_2   + \dots  x_n \cdot A_n  = Y.
    \end{equation} 
    
    For a fixed $j$, it comes that
    
    \begin{equation}
         x_1 \cdot a_{1,j}  + x_2 \cdot a_{2,j}   + \dots  x_n \cdot a_{n,j}  = y_j.
    \end{equation} 

    Using the definition of generalized tropical semiring, we have that there exists $h \in \{1,\dots,n\}$ such that $x_h  a_{h,j} = y_{j}$.

\medskip

Conversely, let us suppose now that $X$ verifies both conditions. Then 

\begin{equation}
     \sum_{i} x_i \cdot a_{i,j} =  x_1 \cdot a_{1,j} + \dots x_{h-1} \cdot a_{{h-1},j} +x_{h} \cdot a_{h,j}+ x_{h+1} \cdot a_{{h+1},j} + \dots + x_n \cdot a_{n,j} 
\end{equation}

Now, by condition \ref{dosLatino}, there exists $h \in \{1,\dots,n\}$ such that $x_h a_{h,j} = y_j$, so we can rewrite the equation as 
\begin{equation}
     \sum_{i} x_i \cdot a_{i,j} = \sum_{i\not = j} x_i \cdot a_{i,j} + y_j
\end{equation}

As  $x_{i}a_{i,j} + y_j = y_j$ for all $j$, and as the semiring is additively idempotent, we finally get 

\begin{equation}
     \sum_{i} x_i \cdot a_{i,j} =  y_j
\end{equation}

for all $j \in \{1,\dots,m\}$. Thus  $X$  is solution of the system.
\end{proof}

\begin{cor}\label{CorCharSol}
Let $(R,+,\cdot)$ be a generalized tropical semiring such that $(R,\cdot)$ is a group. Then $X = (x_1, x_2,...,x_n)$ is solution of the system $XA=Y$ if and only if 

\begin{enumerate}
    \item $x_i \in W_i \ \forall i =1,\dots,n$ \label{unoArabe}.
    \item $\forall j =1,\dots,m$ $\exists h \in \{1,\dots,n\}$ such that $x_h=C_h =  y_{j}a_{h,j} ^{-1}$  \label{dosArabe}
\end{enumerate}  

\noindent where  $a_{h,j} ^{-1}$ is the inverse of $a_{h,j}$ in a generalized tropical semiring having inverses with respect to $\cdot$ and that contains $R$. 

\end{cor}

\begin{proof}

It is enough to show that the conditions are equivalent to those of Theorem \ref{ThmChar}

Firstly, note that the first condition and condition \ref{unoLatino} of Theorem \ref{ThmChar} are equivalent. 

We will show now that, if condition \ref{unoArabe} is true, then condition \ref{dosArabe} is equivalent to condition \ref{dosLatino} of Theorem \ref{ThmChar}.

If \ref{unoArabe} is satisfied, then $x_i \leq C_i=max\ W_i$. In addition, if \ref{dosLatino} of Theorem \ref{ThmChar} is verified, then 

\begin{equation}
    x_h = y_j a_{h,j} ^{-1}\geq \min_{p} \{ y_p a_{h,p}^{-1} \} = C_h \geq x_h.
\end{equation}

\noindent using Corollary \ref{caracterizacionMaximo} and thus $x_h=C_h$. The converse is trivial. 
\end{proof}

\begin{rem} \label{RCaracterizacion}
Let $R$  be a generalized tropical semiring as above and let us consider the system $AX=Y$. $X = (x_1, x_2,...,x_n)$ is a solution of the system if and only if, for every equation  $j \in \{1,\dots,m\}$ of the system:

\begin{enumerate}
    \item $ a_{i,j}\cdot x_i + y_{j} = y_{j}$
    \item $\exists h  \in \{1,\dots,n\}$ such that $ a_{h,j} \cdot x_h = y_{j}$
\end{enumerate}  
\end{rem}

Then, by the previous Corollary,  for every $j =1,\dots,m$ there exists $h \in \{1,\dots,n\}$ such that $C_h =  y_{j}a_{h,j} ^{-1}$ and in addition $x_h=C_h$. As a result, there exits a non-empty set $Index(j) = \{i  \in \{1,\dots,n\}:C_{i} = y_i  a_{i,j}^{-1}\}$. This induces the following result that provides and algorithm to solve linear equations systems.

\begin{thm}
    Let $(R,+,\cdot)$ be a generalized tropical semiring such that $(R,\cdot)$ is a group and let $XA=Y$ be a system of equations with $Y\in R^m$ and $A=(a_{i,j}) \in Mat_{n\times m} (R)$. Determining all the solutions of the system has a computational cost of $o(nm)$.
\end{thm}

\begin{proof}
We observe that the solution of the system is given by the vectors $X=(x_1,x_2, \dots, x_n)$ with the following assignation:

    For every $j$, we can choose $h \in Index(j)$ and such that $x_h= C_{h}$. The rest of the $x_p$ $p=1,\dots ,n$, $p\not=h$ verify that $x_p \leq C_p$.

To prove it, note that every  $X=(x_1,\dots,x_n)$ with this assignation, verifies that $x_i \leq C_i$, and from Lemma \ref{CiWi}, we have $x_i \in W_i$. Moreover we observe that for every $j$, we take some $h \in Index(j)$ with  $x_h= C_{h} =  y_{j}a_{h,j} ^{-1}$. Thus, by the preceding Corollary, it is a solution.

    On the other hand, if $X=(x_1,\dots,x_n)$ is a solution, then it satisfies the conditions of the preceding Corollary. Then $x_i \in W_i$, and so $x_i \leq C_i= \max W_h$. Moreover, for every $j= 1, \dots, m$, there exists $h $ such that $x_h = C_h$. But then $x_h = y_h  a_{h,j}^{-1}$ and hence $h \in Index(j)$. 

    As a result, to determine all the solutions of the system, it is enough to compute $C_i=\max W_i$ for every $i=1,\dots,n$ and $Index(j) = \{i  \in \{1,\dots,n\}:C_{i} = y_i  a_{i,j}^{-1}\}$ for every $j=1,\dots,m$. 
    
    To calculate them, we can use the following algorithm.

    We first compute the matrix $M \in Mat(R)_{n\times m}$, whose $i-th$ column, $M_i$, is of the form 

        \begin{equation}
            M_i=(y_j a_{i,j}^{-1})_{j=1,\dots,m} = \begin{pmatrix}
y_1a_{i1}^{-1} \\
y_2a_{i2}^{-1} \\
\vdots \\
y_ma_{im}^{-1} \\
\end{pmatrix} 
        \end{equation}
\noindent for every $i=1,\dots,n$

        This results in the computation of $nm$ inverses, and  $nm$  operations in the set $R$.

        Then we calculate $C_i=\min M_i = \min_{j=1,\dots,m} y_j a_{i,j}^{-1}$ for every $i=1,\dots,n$, and simultaneously we compute the set $Row(i)=\{j\in \{1,\dots,m\}; C_i=y_j a_{i,j}^{-1} \}$, which give $m$ comparisons for each column, and thus, $nm$ operations.

        Next we build $Index(j)= \{i  \in \{1,\dots,n\}:C_{i} = y_i  a_{i,j}^{-1}\}$ using $Row(i)$ with the following process
        \begin{enumerate}
            \item Take $Index(j)$ empty for every $l$ $j=1,\dots ,m$.
            \item Go over $Row(i)$ for $i=1,\dots,n$ and if $j\in Row(i)$, then add $i$ to $Index(j)$.
        \end{enumerate}

        This process requires to go over $Row(i)$, and since $Row(i) \subseteq  \{1,\dots,m\}$, this procedure require  $o(m)$ comparisons for every $i=1,\dots,n$, and hence the cost is $o(nm)$.

Taking into account that the comparison of two element is made through the addition of both elements in $R$, the total cost is $o(nm)$ basic operations in the ring $(R,+,\cdot)$.
\end{proof}

\begin{rem}
We recall that when the generalized tropical semiring $R$ is such that $(R, \cdot)$ is cancellative, which is less restrictive than being a group, using Theorem \ref{ThmEmb} we can embed this semiring into a generalized tropical semiring $S$ in the conditions of Corollary \ref{CorCharSol}, and then we can solve any linear system as previously, obtaining each solutions in $S$ and then check if any of them is in fact contained in $R$.
\end{rem}

Within the references,  in \cite{Olia21} a method to solve a system of equation by normalization is presented. In \cite{Estructura}, the structure of solution of a system of equation over a tropical semiring is studied, using the rank over rows and columns, and subsequently a generalized method of Cramer is used to find the maximal solution over a tropical semiring.

Examples of system of equations appear in both papers; in \cite{Olia21} a solution, not necessarily maximal, is computed, and in \cite{Estructura}, the authors provide the range of freedom of the solution. Now, using the preceding, we can show the complete set of solutions of those systems of equations.

To avoid misreading of the operation over $R$ as usual ring and as tropical semiring, the operation $(R,\max, +)$ will be denoted as $(R,+_T,\cdot_T)$.

\begin{exmp}
In \cite{Estructura} , the authors compute a solution of the system 

   \begin{equation}
       \begin{bmatrix}
-4 & 7 & 12 & -3 & 0 \\
3 & 2 & 8 & 3 & -1 \\
-9 & 1 & 6 & 0 & 2 \\
2 & 8 & -5 & 1 & -3 \\
\end{bmatrix}
\begin{bmatrix}
x_1 \\
x_2 \\
x_3 \\
x_4 \\
x_5 \\
\end{bmatrix}
=
\begin{bmatrix}
14 \\
10 \\
8 \\
11 \\
\end{bmatrix}.
   \end{equation} 

Let us determine the complete set of solutions of the system as well as the maximal solution. 

Firstly, we calculate $y_j \cdot_T a_{i,j}^{-1} = y_j - a_{i,j}$, obtaining the matrix 

\begin{equation}
(y_j-a_{i,j})_{i,j}=
\begin{bmatrix}
18 & 7 & 2 & 17 & 14 \\
7 & 8 & 2 & 7 & 11 \\
17 & 7 & 2 & 8 & 6 \\
8 & 3 & 16 & 12 & 14\\
\end{bmatrix}.
\end{equation}

Then we have $C_i = \min_{j} \{ y_j  - a_{i,j} \}$ and the rows where these minimum are reached

\begin{equation}
\begin{tabular}{|c|c|c|}
$\max W_i$ & Value & Row\\
\hline
$C_1$ & $7$ & $\{2 \}$ \\
$C_2$ & $3$ & $\{4 \}$ \\
$C_3$ & $2$ & $\{ 1,2,3\}$ \\
$C_4$ & $7$ & $\{2 \}$ \\
$C_5$ & $6$ & $\{ 3\}$ \\
\hline
\end{tabular}
\end{equation}
Now, let us compute $Index(j)$

\begin{equation}
\begin{tabular}{|c|c|c|}
 & Column \\
\hline
$Index(1)$ & $\{ 3\}$  \\
$Index(2)$ & $\{ 1,3,4\}$  \\
$Index(3)$ & $\{ 3,5\}$ \\
$Index(4)$ & $\{ 2\}$  \\
\hline
\end{tabular}
\end{equation}

Thus the solutions are

\begin{equation}
\begin{split}
     &v_1=(7,3 ,2, 7 , h_5) \\
    &v_1=( 7,3 ,2, h_4 , 6) \\
    &v_1=( h_1,3 ,2, 7 , h_5) \\
    &v_1=( h_1,3 ,2, h_4 , 6) \\
    &v_1=( h_1,3, 2, 7 , h_5) \\
    &v_1=( h_1,3, 2, 7 , 6) \\
\end{split}
\end{equation},

with $h_i \leq C_i$.

Moreover, we can observe that the maximal solution is $(7,3,2,7,6)$.

\end{exmp}

\begin{exmp}
    In \cite{Olia21},  the authors compute the maximal solution of 
    \begin{equation}
        \begin{bmatrix}
165 & 57 & 72 & -7 & 0 \\
141 & 64 & 48 & 3 & -1 \\
137 & 101 & 46 & 0 & 2 \\
-243 & 98 & -206 & 156 & -5 \\
\end{bmatrix}
\begin{bmatrix}
x_1 \\
x_2 \\
x_3 \\
x_4 \\
x_5 \\
\end{bmatrix}
=
\begin{bmatrix}
102 \\
78 \\
76 \\
160 \\
\end{bmatrix}
    \end{equation}

    Using our proposed method, we will complete all the solutions, as well as the maximal one.  

\begin{equation}
  (y_j-a_{i,j})_{i,j}= 
  \begin{bmatrix}
-63 & 45 & 30 & 109 & 102 \\
-63 & 14 & 30 & 75 & 79 \\
-61 & -25 & 30 & 76 & 74 \\
403 & 62 & 366 & 4 & 165 \\
\end{bmatrix}
\end{equation}.

Then we have $C_i = \min_{j} \{ y_j  - a_{i,j} \}$ and the rows where those minimum are reached

\begin{equation}
\begin{tabular}{|c|c|c|}
$\max W_i$ & Value & Row\\
\hline
$C_1$ & $-63$ & $\{3 \}$ \\
$C_2$ & $-25$ & $\{3 \}$ \\
$C_3$ & $30$ & $\{ 1,2,3\}$ \\
$C_4$ & $4$ & $\{4 \}$ \\
$C_5$ & $74$ & $\{ 3\}$ \\
\hline
\end{tabular}
\end{equation}.

Now we compute $Index(j)$

\begin{equation}
\begin{tabular}{|c|c|c|}
 & Columns \\
\hline
$Index(1)$ & $\{ 3\}$  \\
$Index(2)$ & $\{ 3\}$  \\
$Index(3)$ & $\{ 1,2,3,5\}$ \\
$Index(4)$ & $\{ 4\}$  \\
\hline
\end{tabular}
\end{equation}

Thus the solution are

\begin{equation}
\begin{split}
    &v_1=(-63,h_2 ,30, 4 , h_5) \\
   &v_1=(h_1,-25 ,30, 4 , h_5) \\
   &v_1=(h_1,h_2 ,30, 4 , h_5) \\
   &v_1=(h_1,h_2 ,30, 4 , 74) \\
\end{split}
\end{equation},

where $h_i \leq C_i$

In addition, the maximal solution is $(-63,-25,30,4,74)$, which matches the one obtained in the original paper.
    
\end{exmp}

\section{Additively idempotent finite semirings}

This section is devoted to study the finite case, where both the existence of solution of linear system and the computation of the maximal solution are guaranteed precisely by the finiteness condition. We will also show how the provided algorithm to compute the solution of these systems has cryptographic consequences.  

\begin{thm}
    Let $R$ be an additively idempotent finite semiring, and let  $XA=Y$ be a system of equations, with $Y\in R^m$ and $A=(a_{i,j}) \in Mat_{n\times m} (R)$. Then, the system is compatible,  $W_i = \{ x \in R : x \cdot A_i + Y = Y \}$ is finite and 

    \begin{equation}
        X=(x_1,\dots,x_n) \text{ such that } x_i = \sum_{x \in W_i} x 
    \end{equation}

\noindent is the maximal solution of the system.
\end{thm}

\begin{proof}
To show it, it is enough to prove that the set $W_i = \{ x \in R : x A_i + Y = Y \}$ has a maximum, which is 

    \begin{equation}
         x_i = \max_{x \in W_i} x = \sum_{x \in W_i} x
    \end{equation}

    and then apply Theorem  \ref{SolucionSistemaGeneral}. 

Given that $W_i$ is finite for every $i=1,\dots ,n$, we can assert that $x_i=\sum_{x \in W_i} x$ for every $i=1,\dots ,n$ is well defined.

Now if $h \in W_i$, then $h \leq x_i$, for every $i=1,\dots ,n$, given that

\begin{equation}
     x_i + h = \sum_{x \in W_i} x + h = \sum_{x \in W_i, x \not = h} x + h + h = \sum_{x \in W_i, x \not = h} x + h =  \sum_{x \in W_i} x = x_i
\end{equation}

Finally, if $a + y = b + y = y$, then $ (a+b) + y = y$, which shows that $W_i$ is additively closed and hence $x_i = \sum_{x \in W_i} \in W_i$.
\end{proof}

\subsection{A cryptographic application}

In \cite{Rosenthal07} the authors introduce a key exchange protocol over a finite semisimple ring $R$ is introduced. The protocol is defined by the following steps. Let $M_1,S,M_2 \in Mat_n (R)$ with $R$ be public. Then Alice chooses two polynomials whose coefficients are in the center of $R$, that we can estimate of degree less or equal $m$, $p_a(x)$,$q_a(x) \in Z[R]_{m}[x]$, and Bob does the same with $p_b(x)$,$q_b(x)\in Z[R]_{m}[x]$ keeping them private.

Next, Alice and Bob exchange the values  $A=p_{a} (M_1) S q_{a}(M_2)$ and $B=p_{b} (M_1) S q_{b}(M_2)$. The shared key is given by 

\[ p_{a}(M_{1}) B q_{a}(M_{2})=p_{a}(M_{1}) p_{b}(M_{1}) S q_{b}(M_{2}) q_{a}(M_{2})=p_{b}(M_{1}) A q_{b}(M_{2})\]

If we denote by $C[A]=\{H \in Mat(R)_n: \exists p(x) \in Z[R][x] \text{ with } p(A)=H\}$, then we can get the shared key by solving the following problem:

\begin{prob} \label{problem}
Given $M_1,M_2,S, A\in Mat_{n}(R) $, with $A \in C[M_1]SC[M_2]$ , find $U\in C[M_1] $ and $V\in C[M_2]$ such that $USV=A$. 
\end{prob}

Note that  

\begin{equation}
\begin{split}
    & A= USB \text{ with } U\in C[M_1], \ V\in C[M_2] \Leftrightarrow A=p(M_1)Sq(M_2) \text{ with }p,q \in Z[R][x] \Leftrightarrow\\
    & \Leftrightarrow A=(\sum_{i=0}^n p_iM_1^i) S (\sum_{j=0}^m q_jM_2^j) \text{ with }p_i,q_i \in Z[R], n,m \in \mathbf{N} \Leftrightarrow\\
    & \Leftrightarrow A=\sum_{i=0}^n (p_iM_1^i S (\sum_{j=0}^m q_jM_2^i)) = \sum_{i=0}^n \sum_{j=0}^m p_iM_1^i S  q_jM_2^j = \sum_{i=0}^n \sum_{j=0}^m p_i q_j M_1^i S  M_2^j 
\end{split}
\end{equation}

This problem is equivalent to solve the system 

\begin{equation}
    A=\sum_{i=0}^n \sum_{j=0}^m p_i q_j M_1^i S  M_2^j 
\end{equation}.

We can observe that the solutions to this system are solutions of the following matrix system, which defines a linear system of equation over the ring $R$

\begin{equation} \label{sistema}
    A=\sum_{i=0}^n \sum_{j=0}^m d_{i,j} M_1^i S  M_2^j.
\end{equation}

Suppose that we solve the system, and we find the values $d_{i,j} \in Z[R]$, determining a solution. Then we can build the function $F[X,Y,Z] = \sum_{i=0}^{k-1} \sum_{j=0}^{k-1} d_{i,j} X^{i} Y Z^{j}$, which verifies that 

\[F[M_1,S,M_2] = \sum_{i=0}^{k-1} \sum_{j=0}^{k-1} d_{i,j} M_{1}^{i} S M_{2}^{j} = A.\]

Then
\begin{equation}\label{sharedkey}
    \begin{split}
        F[M_{1},B,M_{2}] & = \sum_{i =0}^{k-1} \sum_{j =0}^{k-1} d_{i,j} M_{1}^i B M_{2}^j = \sum_{i =0}^{k-1} \sum_{j =0}^{k-1} d_{i,j} M_{1}^i p_{b}(M_{1}) S q_{b}(M_{2}) M_{2}^j = \\ & = \sum_{i =0}^{k-1} \sum_{j =0}^{k-1} d_{i,j} p_{b}(M_{1}) M_{1}^i  S  M_{2}^j q_{b}(M_{2})= p_{b}(M_{1})\sum_{i =0}^{k-1} \sum_{j =0}^{k-1} d_{i,j} M_{1}^i S M_{2}^j  q_{b}(M_{2}) \\
& = p_{b}(M_{1}) F[M_{1},M_{2},S] q_{b}(M_{2}) =p_{b}(M_{1}) A q_{b}(M_{2}) = \\ & = p_{b}(M_{1}) p_{a}(M_{1}) S q_{a}(M_{2}) q_{b}(M_{2}) .
    \end{split}
\end{equation}

In \cite{Rosenthal07}, the authors propose the following finite simple semiring given by

\begin{center}
    \begin{tabular}{c|cccccc}
+ & 0 & 1 & 2 & 3 & 4 & 5 \\
\hline
0 & 0 & 1 & 2 & 3 & 4 & 5 \\
1 & 1 & 1 & 1 & 1 & 1 & 5 \\
2 & 2 & 1 & 2 & 1 & 2 & 5 \\
3 & 3 & 1 & 1 & 3 & 3 & 5 \\
4 & 4 & 1 & 2 & 3 & 4 & 5 \\
5 & 5 & 5 & 5 & 5 & 5 & 5 \\
\end{tabular}
\end{center}

\begin{center}
    \begin{tabular}{c|cccccc}
$\cdot$ & 0 & 1 & 2 & 3 & 4 & 5 \\
\hline
0 & 0 & 0 & 0 & 0 & 0 & 0 \\
1 & 0 & 1 & 2 & 3 & 4 & 5 \\
2 & 0 & 2 & 2 & 0 & 0 & 5 \\
3 & 0 & 3 & 4 & 3 & 4 & 3 \\
4 & 0 & 4 & 4 & 0 & 0 & 3 \\
5 & 0 & 5 & 2 & 5 & 2 & 5 \\
\end{tabular}
\end{center}

Then they proposed to use polynomials of degree 50, and $20 \times 20$ square matrices given by:

\begin{equation}
     M_1=\left( \begin{array}{cccccccccccccccccccc}
1&0&0&0&0&0&0&0&0&0&0&0&0&0&0&0&0&0&0&0\\
0&0&1&0&0&0&0&0&0&0&0&1&0&0&0&0&0&0&0&0\\
0&0&0&1&0&0&0&0&0&0&0&0&0&0&0&0&0&0&0&0\\
0&0&0&0&1&0&0&0&0&0&0&0&0&0&0&0&0&0&0&0\\
0&1&0&0&0&0&0&0&0&0&0&0&0&0&0&0&0&0&0&0\\
0&0&0&0&0&0&1&0&0&0&0&0&0&0&0&0&0&0&0&0\\
0&0&0&0&0&0&0&2&0&0&0&0&0&0&0&0&0&0&1&0\\
0&0&0&0&0&1&0&0&0&0&0&0&0&0&0&0&0&0&0&0\\
0&0&0&0&0&0&0&0&0&1&0&0&0&0&0&0&0&0&0&0\\
0&0&0&0&0&0&0&0&0&0&1&0&0&0&0&0&0&0&0&0\\
0&0&0&0&0&0&0&0&0&0&0&2&0&0&0&0&0&0&0&0\\
0&0&0&0&0&0&0&0&0&0&0&0&1&0&0&0&0&0&0&0\\
0&0&0&0&0&0&0&0&1&0&0&0&0&0&0&0&0&0&0&0\\
0&0&0&0&0&0&0&0&0&0&0&0&0&0&1&0&0&0&0&0\\
0&0&0&0&0&0&0&0&0&0&0&0&0&0&0&1&0&0&0&0\\
0&0&0&0&0&0&0&0&0&0&0&0&0&0&0&0&1&0&0&0\\
0&0&0&1&0&0&0&0&0&0&0&0&0&0&0&0&0&1&0&0\\
0&0&0&0&0&0&0&0&0&0&0&0&0&0&0&0&0&0&1&0\\
0&0&0&0&0&0&0&0&0&0&0&0&0&0&0&0&0&0&0&1\\
0&0&0&0&0&0&0&0&0&0&0&0&0&1&0&0&0&0&0&0 
\end{array} \right)
\end{equation}

\noindent and 

\begin{equation}
    M_2=\left( \begin{array}{cccccccccccccccccccc}
0&0&0&0&0&0&0&0&0&0&0&0&0&0&0&0&0&0&1&0\\
0&0&0&0&0&0&0&0&0&0&0&1&0&0&0&0&0&0&0&0\\
0&0&0&0&0&0&1&0&0&0&0&0&0&0&0&0&0&0&0&0\\
0&0&1&0&0&0&0&0&0&0&1&0&0&0&0&0&0&0&0&0\\
0&0&0&0&0&0&0&0&0&0&0&0&0&0&0&0&0&0&0&4\\
0&0&0&0&0&0&0&0&0&0&0&0&0&0&0&1&0&0&0&0\\
0&2&0&0&0&0&0&0&0&0&0&0&0&0&0&0&0&0&0&0\\
0&0&0&0&0&0&0&0&0&0&0&0&0&0&0&0&0&1&0&0\\
0&0&0&1&0&0&0&0&1&0&0&0&0&0&0&0&0&0&0&0\\
0&0&0&0&0&0&0&0&0&0&0&3&1&0&0&0&0&0&0&0\\
0&0&0&0&0&0&0&0&0&0&0&0&0&0&2&0&0&0&0&0\\
0&0&0&1&0&0&0&0&0&0&0&0&0&0&0&0&0&1&0&0\\
0&0&0&0&0&0&0&0&0&0&1&0&0&0&0&0&0&0&0&0\\
0&0&0&0&0&1&0&0&0&0&0&0&0&0&0&0&0&0&0&0\\
0&0&0&0&0&0&0&0&0&1&0&0&0&0&0&0&0&0&0&0\\
0&0&0&0&0&0&0&1&0&0&0&0&0&0&0&0&0&0&0&0\\
1&0&0&0&0&0&0&0&0&0&0&0&0&0&0&0&0&0&0&0\\
0&0&0&0&1&0&0&0&0&0&0&0&0&0&0&0&0&0&0&0\\
0&0&0&0&0&0&2&0&0&0&0&0&0&0&0&0&1&0&0&0\\
0&0&0&0&0&0&0&0&0&0&0&0&0&1&0&0&0&0&0&0 
\end{array} \right)
\end{equation}

As $S$ they define the matrix

\begin{equation}
   S= \left( \begin{array}{cccccccccccccccccccc}
0&1&0&0&0&0&0&0&0&0&0&0&0&0&0&0&0&0&0&0\\
0&0&1&0&0&0&0&0&0&0&0&1&0&0&0&0&0&0&0&0\\
1&0&0&0&0&0&0&0&0&0&0&0&0&0&0&0&0&0&0&0\\
0&0&0&1&0&1&0&0&0&0&0&0&0&0&0&0&0&0&0&0\\
0&0&0&0&0&0&0&1&0&0&0&0&0&0&0&0&1&0&0&0\\
0&0&0&0&0&0&1&0&0&0&0&0&0&0&0&0&0&0&0&0\\
0&0&0&0&0&1&0&0&0&0&0&0&0&0&0&0&0&0&0&0\\
0&1&0&0&0&0&0&1&0&0&0&0&0&0&0&1&0&0&0&0\\
0&0&0&0&0&0&0&0&0&1&0&0&0&0&0&1&0&2&0&0\\
0&0&0&0&0&0&0&0&0&0&1&0&0&0&0&0&0&0&0&0\\
0&0&0&0&0&0&0&0&0&0&0&1&0&0&0&0&0&0&0&0\\
0&0&0&0&0&0&0&0&1&0&0&0&0&0&0&0&0&0&0&0\\
0&0&0&1&0&0&0&0&0&0&0&0&1&0&0&0&0&0&0&0\\
0&0&0&0&0&0&0&0&0&0&0&0&5&1&0&0&0&0&0&0\\
0&0&1&0&0&0&0&0&0&0&0&0&0&0&1&0&0&0&0&1\\
0&0&0&1&0&0&0&0&0&0&0&0&0&0&0&1&0&0&0&0\\
0&0&0&0&0&0&0&0&0&2&0&0&0&0&0&0&0&0&0&1\\
0&0&2&0&0&0&0&0&0&0&0&0&1&0&0&0&0&0&1&0\\
0&0&0&0&1&0&0&0&0&0&0&0&0&0&0&0&0&1&0&0\\
0&0&0&1&0&0&0&0&0&0&0&0&0&0&0&0&1&0&0&0 
\end{array} \right)
\end{equation}

As a result of the protocol, the matrix $A$ is given by

\begin{equation}
 A= \left( \begin{array}{cccccccccccccccccccc}
0&1&2&2&2&0&2&4&0&2&2&2&2&4&2&4&0&2&0&0\\
1&2&1&1&2&1&1&1&1&1&1&1&1&4&2&1&1&2&1&4\\
1&2&1&1&2&1&1&1&1&1&1&1&1&4&2&1&1&2&1&4\\
1&2&1&1&2&1&1&1&1&1&1&1&1&4&2&1&1&2&1&4\\
1&2&1&1&2&1&1&1&1&1&1&1&1&4&2&1&1&2&1&4\\
1&2&2&1&1&1&2&1&0&2&2&2&5&1&2&1&1&1&1&1\\
1&2&1&1&1&1&1&2&0&2&2&1&5&1&1&2&1&1&1&1\\
1&2&1&1&1&1&1&2&0&2&2&2&1&4&1&1&1&1&1&1\\
0&2&2&2&2&4&2&4&2&2&1&1&2&0&2&0&0&2&0&0\\
0&2&2&2&2&4&2&4&2&2&2&1&2&0&2&2&0&2&0&0\\
0&2&2&2&2&0&2&0&2&2&2&2&2&0&2&2&0&2&0&0\\
0&2&2&1&2&4&2&4&0&1&1&1&1&4&2&1&0&2&0&0\\
0&2&2&2&2&4&2&4&2&1&1&1&2&4&2&1&0&2&0&0\\
1&2&1&1&1&1&1&1&1&1&1&1&5&1&1&1&1&1&1&1\\
1&2&1&1&1&1&1&1&1&2&1&1&5&1&1&1&1&1&1&1\\
1&2&1&1&1&1&1&1&1&1&1&1&1&4&1&1&1&1&1&1\\
1&2&1&1&1&1&1&1&1&1&1&1&5&1&2&1&1&1&1&1\\
1&2&1&1&1&1&1&4&0&2&2&1&5&1&1&4&1&1&1&1\\
1&2&1&1&1&1&1&1&1&2&1&2&5&1&1&1&1&1&1&1\\
1&2&1&1&1&1&1&1&0&2&2&1&5&1&1&1&1&1&1&1 
\end{array} \right)
\end{equation}

Applying the above method, we get the function 

\begin{align*}
F (X,Z,Y)= & X^{1} Z Y^{0}+X^{2} Z Y^{0}+X^{3} Z Y^{0}+X^{4} Z Y^{0}+X^{6} Z Y^{0}+X^{7} Z Y^{0}+X^{9} Z Y^{0}+\\ X^{1} Z Y^{5}+
& X^{2} Z Y^{5}+X^{3} Z Y^{5}+X^{4} Z Y^{5}+X^{6} Z Y^{5}+X^{7} Z Y^{5}+X^{9} Z Y^{5}+X^{1} Z Y^{7}+\\ X^{2} Z Y^{7}+
& X^{3} Z Y^{7}+X^{4} Z Y^{7}+X^{6} Z Y^{7}+X^{7} Z Y^{7}+X^{9} Z Y^{7}+X^{1} Z Y^{8}+X^{2} Z Y^{8}+\\ X^{3} Z Y^{8}+
& X^{4} Z Y^{8}+X^{6} Z Y^{8}+X^{7} Z Y^{8}+X^{9} Z Y^{8}+X^{1} Z Y^{9}+X^{2} Z Y^{9}+X^{3} Z Y^{9}+\\ X^{4} Z Y^{9}+
& X^{6} Z Y^{9}+X^{7} Z Y^{9}+X^{9} Z Y^{9}+X^{1} Z Y^{10}+X^{2} Z Y^{10}+X^{3} Z Y^{10}+X^{4} Z Y^{10}+\\ X^{6} Z Y^{10}+
& X^{7} Z Y^{10}+X^{9} Z Y^{10}+X^{1} Z Y^{11}+X^{2} Z Y^{11}+X^{3} Z Y^{11}+X^{4} Z Y^{11}+X^{6} Z Y^{11}+\\ X^{7} Z Y^{11}+
& X^{9} Z Y^{11}+X^{1} Z Y^{12}+X^{2} Z Y^{12}+X^{3} Z Y^{12}+X^{4} Z Y^{12}+X^{6} Z Y^{12}+X^{7} Z Y^{12}+\\ X^{9} Z Y^{12}+
& X^{1} Z Y^{13}+X^{2} Z Y^{13}+X^{3} Z Y^{13}+X^{4} Z Y^{13}+X^{6} Z Y^{13}+X^{7} Z Y^{13}+X^{9} Z Y^{13}+\\ X^{1} Z Y^{14}+
& X^{2} Z Y^{14}+X^{3} Z Y^{14}+X^{4} Z Y^{14}+X^{6} Z Y^{14}+X^{7} Z Y^{14}+X^{9} Z Y^{14}+X^{1} Z Y^{15}+\\ X^{2} Z Y^{15}+
& X^{3} Z Y^{15}+X^{4} Z Y^{15}+X^{6} Z Y^{15}+X^{7} Z Y^{15}+X^{9} Z Y^{15}+X^{1} Z Y^{16}+X^{2} Z Y^{16}+\\ X^{3} Z Y^{16}+
& X^{4} Z Y^{16}+X^{6} Z Y^{16}+X^{7} Z Y^{16}+X^{9} Z Y^{16}+X^{1} Z Y^{17}+X^{2} Z Y^{17}+X^{3} Z Y^{17}+\\ X^{4} Z Y^{17}+
& X^{6} Z Y^{17}+X^{7} Z Y^{17}+X^{9} Z Y^{17}+X^{1} Z Y^{18}+X^{2} Z Y^{18}+X^{3} Z Y^{18}+X^{4} Z Y^{18}+\\ X^{6} Z Y^{18}+
& X^{7} Z Y^{18}+X^{9} Z Y^{18}+X^{1} Z Y^{19}+X^{2} Z Y^{19}+X^{3} Z Y^{19}+X^{4} Z Y^{19}+X^{6} Z Y^{19}+\\ X^{7} Z Y^{19}+
& X^{9} Z Y^{19}+X^{1} Z Y^{20}+X^{2} Z Y^{20}+X^{3} Z Y^{20}+X^{4} Z Y^{20}+X^{6} Z Y^{20}+X^{7} Z Y^{20}+\\ X^{9} Z Y^{20}+
& X^{1} Z Y^{21}+X^{2} Z Y^{21}+X^{3} Z Y^{21}+X^{4} Z Y^{21}+X^{6} Z Y^{21}+X^{7} Z Y^{21}+X^{9} Z Y^{21}+\\ X^{1} Z Y^{22}+
& X^{2} Z Y^{22}+X^{3} Z Y^{22}+X^{4} Z Y^{22}+X^{6} Z Y^{22}+X^{7} Z Y^{22}+X^{9} Z Y^{22}+X^{1} Z Y^{23}+\\ X^{2} Z Y^{23}+
& X^{3} Z Y^{23}+X^{4} Z Y^{23}+X^{6} Z Y^{23}+X^{7} Z Y^{23}+X^{9} Z Y^{23}+X^{1} Z Y^{24}+X^{2} Z Y^{24}+\\ X^{3} Z Y^{24}+
& X^{4} Z Y^{24}+X^{6} Z Y^{24}+X^{7} Z Y^{24}+X^{9} Z Y^{24}+X^{1} Z Y^{25}+X^{2} Z Y^{25}+X^{3} Z Y^{25}+\\ X^{4} Z Y^{25}+
& X^{6} Z Y^{25}+X^{7} Z Y^{25}+X^{9} Z Y^{25}+X^{1} Z Y^{26}+X^{2} Z Y^{26}+X^{3} Z Y^{26}+X^{4} Z Y^{26}+\\ X^{6} Z Y^{26}+
& X^{7} Z Y^{26}+X^{9} Z Y^{26}+X^{1} Z Y^{27}+X^{2} Z Y^{27}+X^{3} Z Y^{27}+X^{4} Z Y^{27}+X^{6} Z Y^{27}+\\ X^{7} Z Y^{27}+
& X^{9} Z Y^{27}+X^{1} Z Y^{28}+X^{2} Z Y^{28}+X^{3} Z Y^{28}+X^{4} Z Y^{28}+X^{6} Z Y^{28}+X^{7} Z Y^{28}+\\ X^{9} Z Y^{28}+
& X^{1} Z Y^{29}+X^{2} Z Y^{29}+X^{3} Z Y^{29}+X^{4} Z Y^{29}+X^{6} Z Y^{29}+X^{7} Z Y^{29}+X^{9} Z Y^{29}+\\ X^{1} Z Y^{30}+
& X^{2} Z Y^{30}+X^{3} Z Y^{30}+X^{4} Z Y^{30}+X^{6} Z Y^{30}+X^{7} Z Y^{30}+X^{9} Z Y^{30}+X^{1} Z Y^{31}+\\ X^{2} Z Y^{31}+
& X^{3} Z Y^{31}+X^{4} Z Y^{31}+X^{6} Z Y^{31}+X^{7} Z Y^{31}+X^{9} Z Y^{31}+X^{1} Z Y^{32}+X^{2} Z Y^{32}+\\ X^{3} Z Y^{32}+
& X^{4} Z Y^{32}+X^{6} Z Y^{32}+X^{7} Z Y^{32}+X^{9} Z Y^{32}+X^{1} Z Y^{33}+X^{2} Z Y^{33}+X^{3} Z Y^{33}+\\ X^{4} Z Y^{33}+
& X^{6} Z Y^{33}+X^{7} Z Y^{33}+X^{9} Z Y^{33}+X^{1} Z Y^{34}+X^{2} Z Y^{34}+X^{3} Z Y^{34}+X^{4} Z Y^{34}+\\ X^{6} Z Y^{34}+
& X^{7} Z Y^{34}+X^{9} Z Y^{34}+X^{1} Z Y^{35}+X^{2} Z Y^{35}+X^{3} Z Y^{35}+X^{4} Z Y^{35}+X^{6} Z Y^{35}+\\ X^{7} Z Y^{35}+
& X^{9} Z Y^{35}+X^{1} Z Y^{36}+X^{2} Z Y^{36}+X^{3} Z Y^{36}+X^{4} Z Y^{36}+X^{6} Z Y^{36}+X^{7} Z Y^{36}+\\ X^{9} Z Y^{36}+
& X^{1} Z Y^{37}+X^{2} Z Y^{37}+X^{3} Z Y^{37}+X^{4} Z Y^{37}+X^{6} Z Y^{37}+X^{7} Z Y^{37}+X^{9} Z Y^{37}+\\ X^{1} Z Y^{38}+
& X^{2} Z Y^{38}+X^{3} Z Y^{38}+X^{4} Z Y^{38}+X^{6} Z Y^{38}+X^{7} Z Y^{38}+X^{9} Z Y^{38}+X^{1} Z Y^{39}+\\ X^{2} Z Y^{39}+
& X^{3} Z Y^{39}+X^{4} Z Y^{39}+X^{6} Z Y^{39}+X^{7} Z Y^{39}+X^{9} Z Y^{39}+X^{1} Z Y^{40}+X^{2} Z Y^{40}+\\ X^{3} Z Y^{40}+
& X^{4} Z Y^{40}+X^{6} Z Y^{40}+X^{7} Z Y^{40}+X^{9} Z Y^{40}+X^{1} Z Y^{41}+X^{2} Z Y^{41}+X^{3} Z Y^{41}+\\ X^{4} Z Y^{41}+
& X^{6} Z Y^{41}+X^{7} Z Y^{41}+X^{9} Z Y^{41}+X^{1} Z Y^{42}+X^{2} Z Y^{42}+X^{3} Z Y^{42}+X^{4} Z Y^{42}+\\ X^{6} Z Y^{42}+
& X^{7} Z Y^{42}+X^{9} Z Y^{42}+X^{1} Z Y^{43}+X^{2} Z Y^{43}+X^{3} Z Y^{43}+X^{4} Z Y^{43}+X^{6} Z Y^{43}+\\ X^{7} Z Y^{43}+
& X^{9} Z Y^{43}+X^{1} Z Y^{44}+X^{2} Z Y^{44}+X^{3} Z Y^{44}+X^{4} Z Y^{44}+X^{6} Z Y^{44}+X^{7} Z Y^{44}+\\ X^{9} Z Y^{44}+
& X^{1} Z Y^{45}+X^{2} Z Y^{45}+X^{3} Z Y^{45}+X^{4} Z Y^{45}+X^{6} Z Y^{45}+X^{7} Z Y^{45}+X^{9} Z Y^{45}+\\ X^{1} Z Y^{46}+
& X^{2} Z Y^{46}+X^{3} Z Y^{46}+X^{4} Z Y^{46}+X^{6} Z Y^{46}+X^{7} Z Y^{46}+X^{9} Z Y^{46}+X^{1} Z Y^{47}+\\ X^{2} Z Y^{47}+
& X^{3} Z Y^{47}+X^{4} Z Y^{47}+X^{6} Z Y^{47}+X^{7} Z Y^{47}+X^{9} Z Y^{47}+X^{1} Z Y^{48}+X^{2} Z Y^{48}+\\ X^{3} Z Y^{48}+
& X^{4} Z Y^{48}+X^{6} Z Y^{48}+X^{7} Z Y^{48}+X^{9} Z Y^{48}+X^{1} Z Y^{49}+X^{2} Z Y^{49}+X^{3} Z Y^{49}+\\ X^{4} Z Y^{49}+
& X^{6} Z Y^{49}+X^{7} Z Y^{49}+X^{9} Z Y^{49}+X^{1} Z Y^{50}+X^{2} Z Y^{50}+X^{3} Z Y^{50}+X^{4} Z Y^{50}+\\ X^{6} Z Y^{50}+
& X^{7} Z Y^{50}+X^{9} Z Y^{50}
\end{align*}

that verifies $F[M_{1},S,M_{2}]=A$ and thus, the shared key could be recovered as it is given in (\ref{sharedkey}).

\bibliographystyle{plain}
\bibliography{Article}

\end{document}